\documentclass[sn-mathphys]{sn-jnl}

\theoremstyle{thmstyleone}%
\theoremstyle{thmstyletwo}%

\theoremstyle{thmstylethree}%
\usepackage{graphicx}
\usepackage{dashbox}
 \usepackage{mathtools}
 \usepackage{amsmath}
 \usepackage{amssymb}
 \usepackage{bm}

\usepackage{color, colortbl}
\definecolor{Gray}{gray}{0.9}
\usepackage[first=0,last=9]{lcg}

\usepackage{multicol}
\usepackage{multirow}

\theoremstyle{plain}
\newtheorem{thm}{\theoremname}
\theoremstyle{definition}
\newtheorem{defn}[thm]{\definitionname}
\theoremstyle{remark}

\theoremstyle{plain}

\theoremstyle{plain}

\providecommand{\definitionname}{Definition}
\providecommand{\lemmaname}{Lemma}
\providecommand{\propositionname}{Proposition}
\providecommand{\remarkname}{Remark}
\providecommand{\theoremname}{Theorem}
\usepackage[justification=centering]{caption}

\begin{document}

\title[ ]{Lattice Codes for Lattice-Based PKE}


%
%
%

\author*[1]{\fnm{Shanxiang} \sur{Lyu}}\email{lsx07@jnu.edu.cn}

\author[2]{\fnm{Ling} \sur{Liu}}

\author[3]{\fnm{Cong} \sur{Ling}}
 
\author[1]{\fnm{Junzuo} \sur{Lai}}
 
\author[1]{\fnm{Hao} \sur{Chen}}

\affil*[1]{\orgdiv{College of Cyber Security}, \orgname{Jinan University}, \orgaddress{ \city{Guangzhou}, \postcode{510632},  \country{China}}}

\affil[2]{\orgdiv{ College of Computer Science and
		Software Engineering}, \orgname{Shenzhen University}, \orgaddress{ \city{Shenzhen}, \postcode{518060},  \country{China}}}

\affil[3]{\orgdiv{Department of Electrical and Electronic Engineering}, \orgname{Imperial College London}, \orgaddress{ \city{London}, \postcode{SW7 2AZ},  \country{United Kingdom}}}


\abstract{
	Existing error correction mechanisms in lattice-based public key encryption (PKE) rely on either naive modulation or its concatenation with error correction codes (ECC). This paper shows that lattice coding, as a joint ECC and modulation technique, can substitute the naive modulation in existing lattice-based PKEs to enjoy better correction performance. We begin by modeling the FrodoPKE protocol as a noisy point-to-point communication system, where the communication channel is similar to the additive white Gaussian  noise (AWGN) channel. To employ lattice codes for this special channel that hinges on hypercube shaping, we propose an efficient labeling function that converts between binary information bits and lattice codewords. The parameter sets of FrodoPKE are improved towards either higher security levels or smaller ciphertext sizes. For example, the proposed Frodo-1344-E$_\text{8}$ has a 10-bit classical security gain over Frodo-1344.}




\maketitle

\section{Introduction}\label{sec1}
The impending realization of scalable quantum computers has posed a great challenge for modern public key cryptosystems. As Shor's quantum algorithm \cite{Shor97} can solve  the prime factorization and discrete logarithm problems in polynomial time, conventional public-key cryptosystems based on these problems are no longer secure. Although making a prophesy for 
when we can build a large  
quantum computer is hard, 
we should start preparing the next generation quantum-safe cryptosystem as soon as possible, because 
historical experiences show that deploying modern public key cryptography infrastructures takes a long time.

Reacting to this urgency,  the subject of post-quantum cryptography (PQC) has been systematically developed in the last decade \cite{Regev05,Peikert16decade}. PQC aims to design
{cryptosystems secure against quantum attacks,}  while being {able} to run on a classic computer. From 2016, the National Institute of Standards and Technology (NIST) has initiated a process to solicit, evaluate, and standardize one or more quantum-resistant public-key cryptographic algorithms. The process {revolves } around public
key encryption/key encapsulation mechanism (PKE/KEM)  and digital signature proposals.

Recently NIST has announced four post-quantum cryptography standardization candidates \cite{alagic2022status}: CRYSTALS-Kyber for PKE/KEM,  CRYSTALS-Dilithium, FALCON and SPHINCS+ for digital signatures. As the first three candidates are all based on lattices, it's a great victory of lattice-based cryptography (LBC), which enjoys the following prominent advantages. 
{First, LBC enjoys very strong security proofs based on the hardness of worst-case problems.} 
{Second, LBC implementations are notable for their efficiency compared to other post quantum constructions, primarily due to their inherent linear algebra based matrix or vector operations on integers.}
 Finally, 
LBC constructions offer extended functionality for advanced constructions such as identity-based encryption and  fully-homomorphic encryption (FHE). 

{In lattice-based PKE/KEM, the decrypted messages may not be $100\%$ correct.  As the encryption-decryption process  amounts to the transmission of messages through an additive noise channel, error correction techniques have been either implicitly or explicitly employed to reduce the decryption failure rate (DFR).}
  Moreover, since the adversary may extract the secrets by taking advantages of high DFRs, the DFR of a PKE/KEM scheme has to be extremely small (e.g., smaller than $2^{-128}$ or $2^{-140}$) \cite{SAC/FritzmannPS18,impact_dfr_DAnversVV18}. It is therefore worthwhile to improve the error correction mechanism in lattice-based PKE/KEM, with the hope of obtaining  better trade-off parameters: \begin{itemize}
	\item {Security Strength:} If  the error correction mechanism can increase the noise variance while maintaining a small DFR, then the PKE/KEM scheme has a higher security level.
	\item {Communication Bandwidth:} If the error correction mechanism can reduce the modulus while maintaining a small DFR, then the size of the ciphertext is reduced.  
\end{itemize} 

\subsection{Related Works}
KEMs can simultaneously output a session key together with a ciphertext that can be used to recover the session key. Two major approaches to designing lattice-based KEMs are PKEs (KEMs without reconciliation, see, e.g., \cite{naehrig2017frodokem,earlynewhopesimple,iacr/Poppelen16,laura-2110-e8}) and key exchanges (KEMs with reconciliation, see, e.g., \cite{iacr/Ding12a,DBLP:conf/uss/AlkimDPS16,jin2022compact}). As avoiding the error-reconciliation mechanism brings great simplicity, 
we focus on PKEs in this work.

Most lattice-based PKEs have implicitly employed an error correction mechanism which is referred to as  ``naive modulation''. It represents a 
mapping from a binary string to 
different positions in  
$\left\lbrace 0, 1, \ldots, q-1\right\rbrace$. If the noise amplitude is smaller than the error correction radius, then the decryption is correct. Thus a larger $q$ enables higher error correction capability.
Specifically, Regev's learning with errors (LWE) based PKE scheme \cite{Regev05} modulates $1$-bit information $\left\lbrace 0,1 \right\rbrace $ to  $\left\lbrace 0,q/2 \right\rbrace $. Kawachi et.~al. \cite{pkc/KawachiTX07} extends the PKE scheme to multi-bit modulation. 

In recent years, researchers have realized that (digital) error correction codes (ECC) can be concatenated with modulations to obtain better error correction performance. For instance, the LAC \cite{luxianhuiLAC} PKE employs BCH codes for error correction, which helps to reduce the modulo size $q$ from $12289$ to $251$. The reason behind the small $q$ is that, although the modulation level has minus error correction capability, the induced ECC helps to achieve a smaller DFR.  Other examples can be found in the repetition codes based NewHope-Simple \cite{earlynewhopesimple}, {XE5 based HILA5 \cite{Saarinen17},}
 and the Polar codes based NewHope-Simple \cite{entropy/WangL21a}.
The downside of an extra modern ECC is an increased complexity of the program code and a higher sensitivity to side-channel attacks \cite{DBLP:conf/ccs/DAnversTVV19} (information is obtained through physical channels such as power measurements, variable execution time of the decoding algorithm, etc). 

More importantly, using ECC and modulation in a concatenated manner confines the overall performance of the system, whose deficiencies include less flexible number of encoded bits, and the independent decoding nature of modulation and ECC. Fortunately,  the joint design of ECC and modulation (referred to as ``coded modulation'') has been studied in information theory and wireless communications for a few decades.  Ungerboeck's pioneering work \cite{tit/Ungerboeck82} in the 1980s showed that coded modulation exhibits significant performance gains. Then, Forney \cite{tit/Forney88p1,tit/Forney88a} systematically studied coded modulation schemes from coset codes/lattice codes. A breakthrough in information theory is that Erez and Zamir \cite{Erez2004} show high dimensional random lattice codes can achieve the capacity of additive white Gaussian noise (AWGN) channels. Recent years have also witnessed the use of Polar lattices \cite{tcom/LiuYLW19} and LDPC lattices \cite{tit/SilvaS19} in achieving the capacity of AWGN channels.
 In the language of coset codes,  lattice codes represent an \textit{elegant combination} of linear codes and modulation: if points in the constellation are closed, they are protected by ECC; if points are far away, the information bits are directly mapped to them.

It is noteworthy that deploying lattice codes in  lattice-based PKE is not straightforward, because previous lattice coding literature \cite{BK:Zamir} was considering lattice codes for the physical layer (the transmission power of the codes matters), while the modulo $q$ arithmetic in LBC represents a higher layer. In the past few years,
there have been some works that employ lattice codes in PKEs. In 2016 van Poppelen designed a Leech lattice based PKE \cite{iacr/Poppelen16}, while Saliba et.~al. \cite{laura-2110-e8} designed an $E_8$-lattice-based PKE in 2021. The use of $E_8$ and Leech parallels the celebrated  breakthrough in mathematics in recent years: proving the $E_8$ and Leech lattices offer the best sphere packing density in dimensions $8$ and $24$ \cite{viazovska2017sphere,cohn2017sphere}. 
Unfortunately, the labeling technique is missing in the Leech lattice based PKE \cite{iacr/Poppelen16}, while the labeling technique for $E_8$ in \cite{laura-2110-e8} is nonlinear. In this regard, 
 a general lattice-code based error correction formulation, along with efficient linear labeling, is needed for lattice-based PKEs.
 
\subsection{Contributions}
This paper contributes in the following ways, suggesting the naive modulation in lattice based PKE should be replaced with coded modulation.
\begin{itemize}
	\item  We consider the plain-LWE scheme Frodo  \cite{naehrig2017frodokem} and model it as a communication system, over which the communication channel is akin to the AWGN channel. We show that the error correction performance can be easily improved by replacing the naive modulation with lattice-based coded modulation. 
	 In a similar vein, the ring-based or module-based schemes such as NewHope-Simple \cite{earlynewhopesimple} and Kyber \cite{DBLP:conf/eurosp/BosDKLLSSSS18}, can also resort to lattice-based coded modulation. 
	\item We present a universal and efficient labeling technique for cubic-shaping based lattice codes. Due to the modulo $q$ arithmetic, lattice codes in LBC have to use hypercube shaping, which means the coarse lattice should be a simple integer lattice $q\mathbb{Z}^n$.  Although the number of lattice codewords can be easily identified in hypercube shaping, there seems to be no efficient labeling function available in the literature. In response, we propose
	a labeling function to establish a one-to-one map between the
	binary information bits and the set of lattice vectors.
	For a fine lattice whose Hermite parameter is large, we first rewrite its lattice basis to a rectangular form (the product of a unimodular matrix and a diagonal matrix). The proposed labeling is feasible for a wide range of lattices, such as $D_4$, $E_8$, $BW_{16}$, $\Lambda_{24}$, etc. 
	\item	A unified DFR formula over AWGN channels is derived to analyze the DFR of lattice-code based FrodoPKE. Only the Hermite parameter and the kissing number  of lattices are needed in the DFR formula. Previously the  DFRs were calculated by a computationally intensive case-by-case analysis. Via the DFR formula, better parameter sets for  FrodoPKE are derived, where the $E_8$ or $BW_{16}$ based implementations are particularly attractive: their encoding and decoding procedures are simple, and the modified PKE enjoys either higher security levels or smaller ciphertext sizes. 
%
\end{itemize}


The rest of this paper is organized as follows. Background  about
lattice codes and PKE are reviewed in Section II.
The proposed labeling is introduced  and analyzed in Section III.  
Section IV presents a coset-based lattice decoding formulation, along with the pseudocode of decoding $BW_{16}$. Section V presents the improved parameter sets for FrodoPKE. The last section concludes this paper.

\section{Preliminaries}
\subsection{Lattice Codes and Hypercube Shaping}
\begin{defn}[Lattices]{
		A {rank $n$}  lattice $\Lambda$ is a discrete additive subgroup of $\mathbb{R}^{m}$, $m\geq n$. For simplicity, it is assumed that $m=n$ throughout.}
\end{defn}
 
Based on $n$ linearly independent vectors $\mathbf{b}_{1},\ldots\thinspace,\mathbf{b}_{n}$, $\Lambda$ can be written as
\begin{equation}
\Lambda = \mathcal{L}(\mathbf{B}) =z_1\mathbf{b}_{1}+ z_2\mathbf{b}_{2}+ \cdots + z_n\mathbf{b}_{n},
\end{equation}
where $z_1,\ldots,z_n \in \mathbb{Z}$, and $\mathbf{B}=[\mathbf{b}_{1},\ldots\thinspace,\mathbf{b}_{n}]$ is referred to as a {basis} of $\Lambda$. 

\begin{defn}[Closest Vector Problem]
	Considering a query vector $\mathbf{t}$ and a lattice $\Lambda$, the closest vector problem is to find the closest vector to $\mathbf{t}$ in  $\Lambda$.
\end{defn}

The nearest neighbor quantizer $Q_\Lambda(\cdot)$ denotes a function that solves CVP, i.e., 
\begin{equation}\label{eq_quan2}
Q_\Lambda(\mathbf{t})=\mathop{\arg\min}_{\mathbf{v} \in \Lambda} \|\mathbf{t}-\mathbf{v}\|. 
\end{equation}
{In case of a tie, (\ref{eq_quan2}) outputs the candidate with the smallest Euclidean norm.}


\begin{defn}[Fundamental region]
A fundamental region $\mathcal{R}_{\Lambda}$  of a lattice $\Lambda$ includes one and only one point of $\Lambda$, and when shifting it to any lattice point, the whole $\mathbb{R}^n$ space is tiled.
\end{defn}
 
The Voronoi region $\mathcal{V}_{\Lambda}$
is a special case of the fundamental region $\mathcal{R}_{\Lambda}$. It denotes the set of   points in $\mathbb{R}^n$  that are closer to the origin than any other lattice points in $\Lambda$, i.e.,
\begin{equation}{
	\mathcal{V}_{\Lambda}=\{\mathbf{y} \in \mathbb{R}^n \mid \| \mathbf{y} \|  \leq \|\mathbf{y}-\mathbf{w}\|,\,\forall\mathbf{w}\in \Lambda\}.}
	\end{equation}
\begin{defn}[Modulo lattice]
	$\left[\mathbf{x}\right] \mod \Lambda$ denotes the quantization error of $\mathbf{x}$ with respect to $\Lambda$:
	\begin{equation}
	\left[\mathbf{x}\right] \mod \Lambda = \mathbf{x} -Q_\Lambda(\mathbf{x}).
	\end{equation}
\end{defn}

\begin{defn}[Nested lattices]
	Two lattices $\Lambda_f$ and $\Lambda_c$ are nested if $\Lambda_c \subset \Lambda_f$. The denser lattice $\Lambda_f$ is called the \textit{fine/coding} lattice, and $\Lambda_c$ is called the \textit{coarse/shaping} lattice.
\end{defn}

%

Lattice codes are the Euclidean space counterpart of linear codes, and they provide a unified framework to describe the coded modulation techniques \cite{tit/Forney88p1,tit/Forney88a}. The inherent structure is a one-level/multi-level binary encoder and  subset partitioning, which can encode more than $n$ information bits to $n$ symbols.

\begin{defn}[Lattice code]
A lattice code $\mathcal{C}(\Lambda_f,\Lambda_c)$  is the finite set of points in $\Lambda_f$ that lie within $\mathcal{R}_{\Lambda_c}$:
\begin{equation}\label{eq:self-encoding}
\mathcal{C}(\Lambda_f,\Lambda_c) = \Lambda_f \cap
\mathcal{R}_{\Lambda_c}.
\end{equation}
\end{defn}


If $\Lambda_c=p \mathbb{Z}^n$, then (\ref{eq:self-encoding}) is called \textit{hypercube shaping}. A $2$-dimensional example is shown in Fig. \ref{fig:showShaping}. The purple points denote $\Lambda_f$, and those points enclosed with black circles denote $\Lambda_c$. The nesting relation is $\Lambda_c=7\mathbb{Z}^2 \subset \Lambda_f \subset \mathbb{Z}^2$. The fundamental region $\mathcal{R}_{\Lambda_c}$ in the example, enclosed by dashed black lines, is the shifted version of $\mathcal{V}_{\Lambda_c}$, i.e., $\mathcal{R}_{\Lambda_c}= \mathcal{V}_{\Lambda_c}+(3,3)$.

{
The information rate (averaged number of encoded bits) per dimension is defined as
\begin{equation}
B= \frac{1}{n} \log_2 \left(   \frac{\mathrm{Vol}(\Lambda_c)}{\mathrm{Vol}(\Lambda_f)}    \right) .
\end{equation}}
 
The Hermite parameter of a lattice, also identified as the coding gain,  is defined as
\begin{equation}
\gamma(\Lambda) = \lambda_1(\Lambda)^2/ \mathrm{Vol}(\Lambda)^{2/n}
\end{equation}
where $\lambda_1(\Lambda)$ denotes the length of {a} shortest non-zero vector in $\Lambda$,  and {$ \mathrm{Vol}(\Lambda)=\vert \det(\mathbf{B}) \vert$} 
denotes the volume of $\Lambda$. 
The   coding gain $\gamma(\Lambda)$ measures the increase in density of $\Lambda$ over the baseline integer
lattice $\mathbb{Z}$ (or $\mathbb{Z}^n$). Note that the   supremum of $\lambda_1(\Lambda)^2/ \mathrm{Vol}(\Lambda)^{2/n}$ over all $n$-dimensional lattices is known as Hermite's constant.

\begin{figure}[t!]
	\center
	\includegraphics[width=0.5\textwidth]{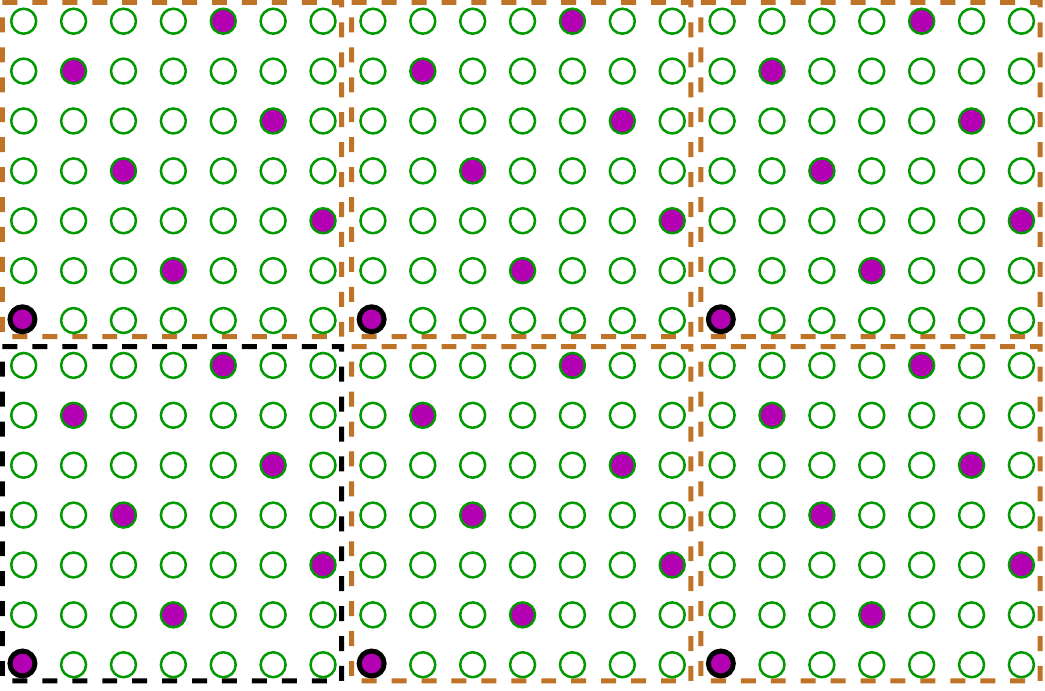}
	\caption{ Demonstration of the lattice code $	\mathcal{C}(\Lambda_f,7\mathbb{Z}^2)=\left\lbrace (0,0), (1,5), (2,3), (3,1), (4,6), (5,4), (6,2) \right\rbrace$} cut from hypercube shaping.
	\label{fig:showShaping}
\end{figure}

\subsection{PKE/KEM in LBC}
  

FrodoKEM \cite{naehrig2017frodokem} is a simple and conservative KEM from generic lattices, and it is
one of two post-quantum algorithms recommended by the German Federal Office for Information Security (BSI) as cryptographically suitable for long-term confidentiality \cite{bsi2020cryptographic}.
The core of FrodoKEM is a public-key encryption scheme called FrodoPKE, whose IND-CPA security is
tightly related to the hardness of a corresponding learning with errors problem. 
{Due to the lack of algebraic structure, the security estimates of FrodoPKE rely on fewer assumptions than} 
other PKE/KEM schemes based on ring or module LWE. 

A public key encryption scheme $\mathsf{PKE}$ is a tuple of algorithms ($\mathsf{KeyGen}$, $\mathsf{Enc}$, $\mathsf{Dec}$) along with a message space $\mathcal{M}$.

In the key generation algorithm, by sampling $\mathbf{S},\mathbf{E} \sim {\chi_{\sigma}^{n'\times \bar{n}}} $, with $\chi_{\sigma}$ being a (truncated) discrete Gaussian distribution with width $\sigma$, and sampling $\mathbf{A}$ from a uniform distribution in $\mathbb{Z}_q^{n'\times n'}$, it computes 
\begin{equation}
\mathbf{B} = \mathbf{AS} + \mathbf{E} \in \mathbb{Z}_q^{n'\times \bar{n}}.
\end{equation}
The public key is $pk=(\mathbf{B},\mathbf{A})$, and the secret key is $sk=\mathbf{S}$.

In the part of public key encryption, it samples $\mathbf{S}', \mathbf{E}' \sim {\chi_{\sigma}^{\bar{m} \times n'}}$, $\mathbf{E}'' \sim {\chi_{\sigma}^{\bar{m}  \times \bar{n}}}$, and 
computes
\begin{align}
\mathbf{C}_1 &=  \mathbf{S}'\mathbf{A} +\mathbf{E}'\\
\mathbf{V}	&=\mathbf{S}'\mathbf{B}  + \mathbf{E}''.
\end{align} 

To encrypt a message $\mu \in \mathcal{M} = \left\lbrace 0,1 \right\rbrace ^{\bar{m}\bar{n}B}$, the ciphertext is generated by  
\begin{equation}
c=	(\mathbf{C}_1,\, \mathbf{C}_2=\mathbf{V}+ \mathsf{Frodo.EncodeM}(\mu)).
\end{equation}
The function $\mathsf{Frodo.EncodeM}$ represents a matrix encoding function of bit strings. In an element-wise manner, each $B$-bit value is transformed into the $B$ most significant bits of the corresponding entry modulo $q$. {We refer to 
	   $\mathsf{Frodo.EncodeM}$ as ``naive modulation'', as it 
	   amounts to a special case of the lattice code based encoding that employs hypercube shaping, with $\Lambda_f=q/(2^{B})\mathbb{Z}^{64}$, $\Lambda_c=q\mathbb{Z}^{64}$.} 

To decrypt, it employs the secret key $\mathbf{S}$ and the ciphertext   $\mathbf{C}_1,\mathbf{C}_2$ to compute
\begin{equation}
\hat{\mu} = \mathsf{Frodo.DecodeM}(\mathbf{C}_2- \mathbf{C}_1 \mathbf{S}),
\end{equation}
where $\mathsf{Frodo.DecodeM}$ standards for the demodulation function. The FrodoPKE protocol is summarized in Fig. \ref{fig:FRODOPKE}. 

When targeting security levels
 1, 3 and 5 in the NIST call for proposals (matching or exceeding the brute-force security of AES-128, AES-192, AES-256), the recommended parameters are
\begin{align*}
&\mathrm{Frodo}\text{-}640:\,  
n'=640, \bar{n}=8, \bar{m}=8, q=2^{15}, \sigma=2.75, \mathcal{M}= \left\lbrace 0,1 \right\rbrace^{128} \\
&\mathrm{Frodo}\text{-}976:\, 
n'=976, \bar{n}=8, \bar{m}=8, q=2^{16}, \sigma=2.3, \mathcal{M}= \left\lbrace 0,1 \right\rbrace^{192}\\
&\mathrm{Frodo}\text{-}1344:\, 
n'=1344, \bar{n}=8, \bar{m}=8, q=2^{16}, \sigma=1.4, \mathcal{M}= \left\lbrace 0,1 \right\rbrace^{256}.
\end{align*}

 
\begin{figure}[t!]
	\centering
	\fbox{
		\begin{tabular}{ccc}
			\multicolumn{3}{c}{Input Parameters: $q$, $n'$, $\bar{n}$, $\bar{m}$,  $\chi_{\sigma}$.}\tabularnewline
			\hline 
			\textbf{Alice} &  & \textbf{Bob}\tabularnewline
			$\mathbf{A}\leftarrow_{\$}\mathbb{Z}_{q}^{n'\times n'}$ &  & \tabularnewline
			$\mathbf{S},\mathbf{E}\leftarrow_{\$}{\chi_{\sigma}^{n'\times\bar{n}}}$ &  & $\mathbf{S}',\mathbf{E}'\leftarrow_{\$}{\chi_{\sigma}^{\bar{m}\times n'}}$\tabularnewline
			$\mathbf{B}=\mathbf{AS+E}$ & $\xrightarrow[]{\mathbf{A},\mathbf{B}}$ & $\mathbf{E}''\leftarrow_{\$}{\chi_{\sigma}^{\bar{m}\times\bar{n}}}$\tabularnewline
			&  & $\mathbf{C}_{1}=\mathbf{S'A+E'}$\tabularnewline
			&  & $\mathbf{V}=\mathbf{S'B+E''}$\tabularnewline
			&  & $\mu \in \left\lbrace 0,1\right\rbrace ^{\bar{m}\bar{n}B}$\tabularnewline
			$\mathbf{Y}=\mathbf{C}_{2}-\mathbf{C}_{1}\mathbf{S}$ & $\xleftarrow[]{\mathbf{C}_1,\mathbf{C}_2}$ & $\mathbf{C}_{2}=\mathbf{V}+\text{\ensuremath{\mathsf{Frodo.EncodeM}}(\ensuremath{\mu})}$\tabularnewline
			$\hat{\mu}=\mathsf{Frodo.DecodeM}(\mathbf{Y})$ &  & \tabularnewline
		\end{tabular}
	}
	\caption{The FrodoPKE protocol.}
	\label{fig:FRODOPKE}
\end{figure}

\section{The Proposed Scheme}
\begin{figure*}[t]
	\center
	
	\includegraphics[width=1 \textwidth]{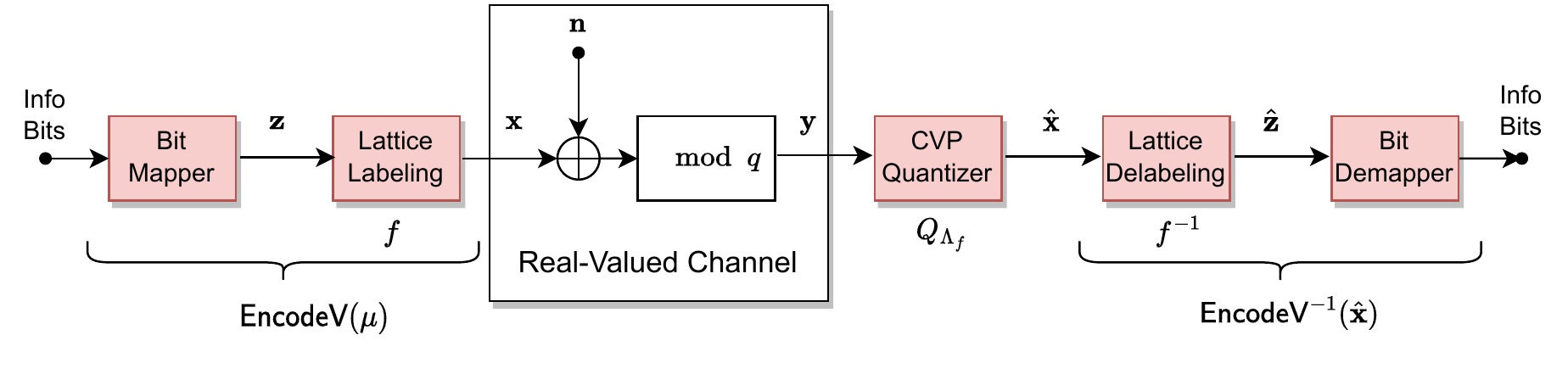}
	\caption{The equivalent communication system model.}
	\label{fig:LATTICEFLOW}
\end{figure*}

\subsection{Equivalent Communication Model} 
Recall that the decryption algorithm of FrodoPKE computes
\begin{align}
\mathbf{Y} &=\mathbf{C}_2 -\mathbf{C}_1\mathbf{S} \notag\\
&=\mathsf{Frodo.EncodeM}(\mu) +\mathbf{S}'\mathbf{E}+\mathbf{E}''-\mathbf{E}'\mathbf{S}, \label{matrixmodel}
\end{align}
whose addition is over the modulo $q$ domain. 
From the perspective of communications, this amounts to transmitting the modulated $\mu$ through an additive noise channel. 
Specifically, Eq. (\ref{matrixmodel}) can be formulated as
\begin{equation}
\mathbf{y}= \mathbf{x} + \mathbf{n} \mod q,
\end{equation}
where  $\mathbf{x}=\mathsf{EncodeV}(\mu)\in \mathbb{R}^{\bar{m}\bar{n}}$ denotes a general error correction function, and $\mathbf{y}$, $\mathbf{n}$ represent the vector form of $\mathbf{Y}$ and $\mathbf{S}'\mathbf{E}+\mathbf{E}''-\mathbf{E}'\mathbf{S}$, respectively. Since the element-wise modulo $q$ is equivalent to modulo a lattice $q\mathbb{Z}^{\bar{m}\bar{n}}$,  $\mathsf{EncodeV}$ can be designed from the perspective of lattice codes. 

{The flowchart of the communication model is plotted in Fig. \ref{fig:LATTICEFLOW}, which contains the following operations:}
\begin{itemize}
	\item \textit{Bit Mapper and Demapper}: The former maps binary information bits to an information vector $\mathbf{z}$ defined over integers. The later performs the inverse operation. These operations are straightforward.
	\item \textit{Lattice Labeling and Delabeling}: Given a message index $\mathbf{z}$, lattice labeling finds its corresponding lattice codeword $\mathbf{x}\in 	\mathcal{C}(\Lambda_f,\Lambda_c=q\mathbb{Z}^{\bar{m}\bar{n}})$. Delabeling denotes the inverse of labeling. 
	\item \textit{CVP Quantizer}: It returns the   closet lattice vector to $\mathbf{y}$ over $\Lambda_f$. The CVP algorithm of 
	$Q_{\Lambda_f}(\cdot)$ will be examined in Section \ref{secCVP}.
\end{itemize}

{While the conventional $\mathsf{Frodo.EncodeM}$ employs   $\Lambda_f=q/(2^{B})\mathbb{Z}^{\bar{m}\bar{n}}$ that leads to  simple labeling functions and CVP quantization, our work seeks to employ a better $\Lambda_f$ for   error correction performance. Thus the associated labeling function and CVP quantization are more involved.} 
 

\subsection{Lattice Labeling and Delabeling} 
{This section will show that for any fine lattice with a basis in a rectangular form, a linear labeling from certain index sets to lattice codewords  can be generically defined.}

\begin{defn}[Rectangular Form]
	A lattice basis $\mathbf{B}$ is in a rectangular form if 
	\begin{equation}\label{eqrec4}
	\mathbf{B}= \mathbf{U} \cdot \mathrm{diag} (\pi_1,\pi_2,\ldots,\pi_n) ,
	\end{equation}
	where $\mathbf{U}\in \mathrm{GL}_n(\mathbb{Z})$ is a unimodular matrix, and $\pi_1,\pi_2,\ldots,\pi_n \in \mathbb{Q}^{+}$.
\end{defn}

For any lattice with a rational basis, it has a rectangular form. Specifically, consider the Smith Normal Form factorization of a lattice basis $\mathbf{B}^* \in \mathbb{Q}^{n\times n}$, then we have 
\begin{equation}
\mathbf{B}^*= \mathbf{U} \cdot \mathrm{diag} (\pi_1,\pi_2,\ldots,\pi_n) \cdot \mathbf{U}',
\end{equation}
where $\mathbf{U}, \mathbf{U}' \in \mathrm{GL}_n(\mathbb{Z})$\footnote{
	{The determinants of $\mathbf{U}$ and $\mathbf{U}'$ are $1$ or $-1$ by incorporating the necessarily rational factors into $\mathrm{diag} (\pi_1,\pi_2,\ldots,\pi_n)$.}}. As lattice bases are equivalent up to unimodular transforms, the term $\mathbf{U}'$ can be canceled out, and the rectangular form is derived.


For a lattice that features a rectangular form,  an efficient labeling scheme can be constructed. The idea is that the combination of rectangular form and non-uniform labeling  amounts to hypercube shaping.  Specifically, let the fine lattice be
\begin{equation}\label{eq_aclf}
\Lambda_f = \mathcal{L}(\mathbf{B}_f) = \mathcal{L}(\mathbf{U} \cdot \mathrm{diag} (\pi_1,\pi_2,\ldots,\pi_n)).
\end{equation}
Let $p\in \mathbb{Z}^{+}$ be a common multiplier of $\pi_1,\pi_2,\ldots,\pi_n$, and define
\begin{equation}
p_1 = p/\pi_1,p_2 = p/\pi_2,\ldots,p_n = p/\pi_n.
\end{equation}
If $\mathbf{B}_c = \mathbf{B}_f \mathrm{diag} (p_1,p_2,\ldots,p_n)$, we have 
\begin{align} 
\Lambda_c &= \mathcal{L}(\mathbf{U} \cdot \mathrm{diag} (\pi_1,\pi_2,\ldots,\pi_n) 
\cdot \mathrm{diag} (p_1,p_2,\ldots,p_n)) \notag \\
&=\mathcal{L}(p\mathbf{U}) \notag\\
&=p\mathbb{Z}^n. \label{eq_sha_3}
\end{align}
The last equality is due to the fact that a unimodular matrix can be regarded as a lattice basis of $\mathbb{Z}^n$. 
Hence modulo $\Lambda_c$ becomes equivalent to modulo $p$. Then we arrive at the following theorem.  

\begin{thm}[Labeling Function]
	Let the message space be
	\begin{equation}
	\mathcal{I} = \left\lbrace 0,1, \ldots, p_1-1 \right\rbrace \times \cdots \times \left\lbrace 0,1, \ldots, p_n-1 \right\rbrace, 
	\end{equation}
	and the pair of nested lattices be $\Lambda_f= \mathcal{L}(\mathbf{B}_f) = \mathcal{L}(\mathbf{U} \cdot \mathrm{diag} (p/p_1,p/p_2,\ldots,p/p_n))$, $\Lambda_c=\mathcal{L}(p\mathbf{U})=p\mathbb{Z}^n$.
	With $\mathbf{z} \in \mathcal{I}$,   the function $f:$ $\mathcal{I} \rightarrow 	\mathcal{C}(\Lambda_f,\Lambda_c)$,
	\begin{equation}\label{eq_f_trans}
	f( \mathbf{z}) =  \left[ \mathbf{B}_f \mathbf{z} \right] \mod p
	\end{equation}
	is bijective.
\end{thm}
\begin{proof}
	It suffices to prove that $f$ is both injective and
	surjective. ``Injective'' means no two elements in the domain of the function gets mapped to the same image, i.e., for $\mathbf{z}_1,\mathbf{z}_2 \in \mathcal{I}$,
	\begin{equation}
	\mathbf{z}_1 \neq \mathbf{z}_2 \rightarrow f(\mathbf{z}_1) \neq f(\mathbf{z}_2).
	\end{equation}
	{We prove this by contradiction, showing $\mathbf{z}_1 \neq \mathbf{z}_2 \rightarrow f(\mathbf{z}_1) = f(\mathbf{z}_2)$ does not hold.}
	 If $f(\mathbf{z}_1)=f(\mathbf{z}_2)$, it implies that we can find $\mathbf{z}_1,\mathbf{z}_2 \in \mathcal{I}$, $\mathbf{z}_3 \in \mathbb{Z}^n$ such that
	\begin{align}
	&\mathbf{B}_f(\mathbf{z}_1-\mathbf{z}_2)=\mathbf{B}_f \cdot \mathrm{diag} (p_1,p_2,\ldots,p_n)\cdot\mathbf{z}_3 \notag \\
	&\rightarrow \mathbf{z}_1-\mathbf{z}_2=  \mathrm{diag} (p_1,p_2,\ldots,p_n)\cdot\mathbf{z}_3. \label{eqprove20}
	\end{align}
	Then (\ref{eqprove20}) has a solution only when $\mathbf{z}_3=\mathbf{0}$, which leads to $\mathbf{z}_1=\mathbf{z}_2$.
	
	``Surjective'' means that any element in the range of the
	function is hit by the function. Recall that the number of {coset representatives} of $\Lambda_f/\Lambda_c$ is
	\begin{equation}
\vert \mathrm{det}(\mathbf{B}_c)\vert/
\mid\mathrm{det}(\mathbf{B}_f)\mid=p_1p_2\cdots p_n.
	\end{equation}
	As $ \vert{\mathcal{I}} \vert =p_1p_2\cdots p_n$, it follows from the injective property that 
	all the {coset representatives} have been  hit distinctively. So the surjection is proved.
\end{proof}

{Denote $\mathbf{x}=f(\mathbf{z})$. The inverse of $f$ is given by 
	\begin{equation}\label{eq21index}
	{\mathbf{z}}=f^{-1}({\mathbf{x}}) \triangleq \mathbf{B}_f^{-1}{\mathbf{x}} \mod (p_1, \ldots, p_n),
	\end{equation}
	which stands for ${z}_i = \left( \mathbf{B}_f^{-1}{\mathbf{x}}\right)_i\,\mod \, p_i, \, i=1, \ldots, n.$ 
	As the labeling and delabeling process also encounters an additive noise channel, we examine the correct recovery condition hereby. Assume that the receiver's side has the noisy observation $\mathbf{x}+\mathbf{n}$, with $\mathbf{x}\in \Lambda_f$ and $\mathbf{n}$ being the additive noise. }
	
	{
	\begin{thm}[Correct Decoding]\label{thm_er}
		If $Q_{\Lambda_{f}}(\mathbf{n})\in \Lambda_{c}$, then $f^{-1}(Q_{\Lambda_{f}}(\mathbf{x}+\mathbf{n}))=f^{-1}(\mathbf{x})$.
	\end{thm}
	\begin{proof}
		Notice that 
		\begin{align}
		Q_{\Lambda_{f}}(\mathbf{x}+\mathbf{n})=\mathbf{x}+	Q_{\Lambda_{f}}(\mathbf{n}),
		\end{align}
		then we have
		\begin{equation}
		f^{-1}(Q_{\Lambda_{f}}(\mathbf{x}+\mathbf{n}))=
		\mathbf{B}_f^{-1}{\mathbf{x}} + \mathbf{B}_f^{-1}	Q_{\Lambda_{f}}(\mathbf{n}) \mod (p_1, \ldots, p_n)
		\end{equation}
		The condition of $Q_{\Lambda_{f}}(\mathbf{n})\in \Lambda_{c}$ implies that this vector of the coarse lattice can be written as $\mathbf{B}_f \mathrm{diag} (p_1,p_2,\ldots,p_n) \mathbf{k}$ for a $\mathbf{k}\in \mathbb{Z}^n$. This yields $Q_{\Lambda_{f}}(\mathbf{n}) \mod (p_1, \ldots, p_n) = \mathbf{0}$ and the theorem is proved.
	\end{proof}
	We summarize three cases for the correct recovery of messages. i) Noiseless: $\mathbf{n}=\mathbf{0}$. ii) Noise is small: $Q_{\Lambda_{f}}(\mathbf{n})=\mathbf{0}$. iii) Noise is large but still in the coarse lattice: $Q_{\Lambda_{f}}(\mathbf{n})\in \Lambda_{c}$.
}

An example of using   labeling and delabeling is given below.

\noindent \textbf{Example:} Consider the  $D_4$ lattice, whose lattice basis and its inverse are respectively given by
\begin{align}
&\mathbf{B}_{D_4} = \left[\begin{matrix} 1&0&0&0\\
0&1&0&0\\
0&0&1&0\\
1&1&1&1\end{matrix}\right]
\cdot \mathrm{diag} (1,1,1,2), \\
&\mathbf{B}_{D_4}^{-1}= \left[\begin{matrix} 1&0&0&0\\
0&1&0&0\\
0&0&1&0\\
-0.5&-0.5&-0.5&0.5\end{matrix}\right].
\end{align}
To encode $7$ bits over $4$ dimensions, 
let the pair of nested lattices be $(\Lambda_f,\Lambda_c)=(D_4, 4\mathbb{Z}^4)$, and the
message space be
\begin{equation}
\mathcal{I}=\left\lbrace 0,1,2,3 \right\rbrace^3 \times \left\lbrace 0,1 \right\rbrace.
\end{equation} 
W.l.o.g, let the input binary string be $\left\lbrace 0,1,1,0,1,1,1 \right\rbrace$. Then the ``Bit Mapper'' transforms the bits to a vector in $\mathcal{I}$:
\begin{equation*}
\mathbf{z}=[1,2,3,1]^\top.
\end{equation*}
By using lattice labeling in Eq. (\ref{eq_f_trans}), we have
\begin{equation*}
\mathbf{x}= f(\mathbf{z})= [1,2,3,0]^\top.
\end{equation*}
In the noiseless case of $\mathbf{n}=0$, we 
have
\begin{align}
	f^{-1}(\mathbf{x}) &=f^{-1}([1,2,3,0]^\top)\\
	&= [1,2,3,-3]^\top \mod (4,4,4,2)\\
  &= [1,2,3,1]^\top.
\end{align}
In the noisy case of $\mathbf{n}=[4,4,4,4]^\top$, we 
have $Q_{\Lambda_{f}}(\mathbf{n})\in \Lambda_{c}$, and 
\begin{align}
f^{-1}(Q_{\Lambda_{f}}(\mathbf{x}+\mathbf{n})) &=f^{-1}([5,6,7,4]^\top)\\
&= [9,10,11,-7]^\top \mod (4,4,4,2)\\
&= [1,2,3,1]^\top.
\end{align}
Finally, the  ``Bit Demapper''   transforms the information integers to bits, which equal to the original input. 

\subsection{Rectangular Forms of Code-Based Lattices}

The proposed labeling is feasible for a wide range of lattices, such as low-dimensional optimal lattices $D_2$, $D_4$, $E_8$, $\Lambda_{24}$, and the general Construction-A and Construction-D lattices. 
Construction A and Construction D are popular techniques of lifting linear codes to lattices, based on which many remarkable lattices with large coding gains have been constructed, such as the Barnes--Wall lattices \cite{tit/Forney88a,tit/SalomonA05,dcc/SalomonA07} and the polar lattices \cite{tcom/LiuYLW19,LiuSL21}.  {Let $C$ be a linear binary code of length $n$, dimension $k$ and minimum distance $d$, denoted as $(n,k,d)$. } 	


\begin{defn}[Construction A \cite{Conway1999}]
A vector $\mathbf{y}$ is a lattice vector of the Construction-A lattice over $C$ if and only if $\mathbf{y}$ modulo $2$ is congruent to a codeword of $C$.
\end{defn}

{Let $\phi(\cdot)$ be a natural mapping function from $\mathbb{F}_2$ to $\mathbb{R}$ with  $\phi(0)=0, \phi(1)=1$ for a scalar input, and $\phi(\cdot)$ is applied element-wise for a vector/matrix input.}
 Let $\mathbf{G}$ be the generator matrix of $C$. By reformulating it as the Hermite normal form of $\left\lbrace \mathbf{I},\mathbf{A} \right\rbrace$, the Construction-A lattice of $C$ can be written as
\begin{equation}
\Lambda_{A} =  \mathcal{L}\left( \left[\begin{matrix} \phi(\mathbf{I})& \mathbf{0}\\ \phi(\mathbf{A}) & 2\mathbf{I}\end{matrix}\right]\right).
\end{equation}
The lattice basis of $\Lambda_{A}$ is therefore of a rectangular form. The volume of $\Lambda_{A}$ is
\begin{equation}\label{eqvolCA}
V(\Lambda_{A}) = 2^{n-k}.
\end{equation}


\begin{defn}[Construction D \cite{Conway1999}]
	Let $C_0 \subset C_1 \subset \cdots \subset C_a = \mathbb{F}_2^n$ be a family of nested binary linear codes, where $C_i$ has parameters $(n,k_i,d_i)$ and $C_a$ is the trivial $(n,n,1)$ code. A vector $\mathbf{y}$ is a lattice vector of the Construction-D lattice over $(C_0, \ldots, C_a)$  if and only if $\mathbf{y}$ is congruent (modulo $2^{a}$) to a vector in $C_0+2C_1+\cdots+2^{a-1}C_{a-1}$.
\end{defn}

Denote the generator matrices of $C_0$, $C_i$, and $C_a$ as
\begin{align}
&	\mathbf{G}_{0}=\left[\begin{array}{cccc}
\mid & \mid &  & \mid\\
\mathbf{g}_{1} & \mathbf{g}_{2} & \cdots & \mathbf{g}_{k_{0}}\\
\mid & \mid &  & \mid
\end{array}\right]\\
&\mathbf{G}_{i}=\left[\begin{array}{cccccc}
\mid & \mid &  & \mid &  & \mid\\
\mathbf{g}_{1} & \mathbf{g}_{2} & \cdots & \mathbf{g}_{k_{0}} & \cdots & \mathbf{g}_{k_{i}}\\
\mid & \mid &  & \mid &  & \mid
\end{array}\right]\\
&\mathbf{G}_{a}=\left[\begin{array}{cccccccc}
\mid & \mid &  & \mid &  & \mid &  & \mid\\
\mathbf{g}_{1} & \mathbf{g}_{2} & \cdots & \mathbf{g}_{k_{0}} & \cdots & \mathbf{g}_{k_{i}} & \cdots & \mathbf{g}_{k_{a}}\\
\mid & \mid &  & \mid &  & \mid &  & \mid
\end{array}\right],\label{eq_gadef}
\end{align}
where $1 \leq k_0 \leq k_1 \leq \cdots \leq k_a =n$.
Then the code formula of a Construction-D lattice is 
\begin{align}
\Lambda_{D} &= \bigcup_{\mathbf{u}_i\in {\left\lbrace 0,1 \right\rbrace }^{k_i}}\left( \sum_{i=0}^{a-1}2^i \phi(\mathbf{G}_i) \mathbf{u}_i \mathbf{}\right) + 2^a  \mathbb{Z}^n \\
&= \mathcal{L}(\phi(\mathbf{G}_a)\cdot\mathrm{diag} (2^0 \mathbf{1}_{k_0}, \ldots, 2^a \mathbf{1}_{k_a-k_{a-1}})), \label{eq_consd28}
\end{align}
where $\mathbf{1}_{k_i}$ denotes an all-one vector of dimension $k_i$, $\phi(\mathbf{G}_{a})$ is a unimodular matrix. Thus $2^a  \mathbb{Z}^n \subset \Lambda_{D}$ and the  volume of a   Construction-D lattice is 
\begin{equation}\label{eqvolCD}
V(\Lambda_{D}) = 2^{an -\sum_{i=0}^{a-1}k_i}.
\end{equation}
By using Construction D over Reed-Muller codes, the Barnes--Wall lattices can be obtained \cite{tit/Forney88a} \footnote{{Barnes--Wall lattices can also be defined recursively \cite[Definition 1.1]{cc/GrigorescuP17}.}}. Some low-dimensional examples are
\begin{align}
&BW_{8} = (8,4,4) + 2\mathbb{Z}^8 \approxeq E_8\\
&BW_{16}  =     (16,5,8) +2(16,15,2) + 4\mathbb{Z}^{16} \approxeq \Lambda_{16} \label{eq_29bw16} \\
&BW_{32}=  (32,6,16) +2(36,26,4) + 4\mathbb{Z}^{32} \\
&BW_{64}= (64,7,32) +2(64,42,8) +4(64,63,2)+ 8\mathbb{Z}^{64},
\end{align}
where $\approxeq$ denotes equality up to rotations. The rectangular-form lattice basis in (\ref{eq_consd28}) can be derived by considering the Kronecker product based construction of Reed-Muller codes \cite[Section I-D]{tit/Arikan09}. With explicit rectangular forms, the lattice bases of $E_8$, $BW_{8}$ and $BW_{16}$ are shown in Appendix A.

%






\subsection{CVP Decoding}\label{secCVP}
Enumeration and sieving are two popular types of CVP algorithms for decoding random lattices \cite{codcry/HanrotPS11,basesearch/Voulgaris11}. For the code-based lattices used in error correction, they feature strong structures, thus algorithms should exploit the structures to improve the decoding efficiency.  {While there exist bounded distance  decoding (BDD) for the considered Barnes--Wall lattices \cite{isit/MicciancioN08,cc/GrigorescuP17}, BDD fails to reach the DFR of CVP decoding.} Exploiting the structure of cosets,  
efficient CVP algorithms of $E_8$ and $D_n$ can be found in \cite{ConwayS82a}. In a similar vein, this section examines the CVP decoding of $BW_{16}$, $BW_{32}$ and $BW_{64}$.

 

\subsubsection{Lattice Partition as Cosets}
A natural and efficient way to design CVP algorithms for Construction-D lattices is to partition the lattice as the union of cosets. If $\Lambda$ equals to the union of $\Lambda'$ cosets, the CVP of $\Lambda$ can resort to that of $\Lambda'$:
{\begin{align}\label{union_decode1}
	Q_{\Lambda }(\mathbf{t}) &=  Q_{\Lambda'+\mathbf{g}'}(\mathbf{t}),\\
	\mathbf{g}' &=\mathrm{argmin}_{\mathbf{g}\in \Lambda/\Lambda'} \left\| \mathbf{t} -	Q_{\Lambda'+\mathbf{g}}(\mathbf{t}) \right\|,\nonumber
	\end{align}}
where $Q_{\Lambda'+\mathbf{g}}(\mathbf{t}) = \mathbf{g} + 	Q_{\Lambda'}( \mathbf{t}-\mathbf{g})$. Denote the number of cosets as $\vert\Lambda/\Lambda'\vert$. Then the computational complexity of $Q_{\Lambda}$ is $\vert\Lambda/\Lambda'\vert$ times larger than $Q_{\Lambda'}$.

All the Construction-D lattices admit a $\mathbb{Z}^n$ based coset partition, but such partition has a huge number of cosets in general. Whenever possible, partitioning the lattice as $D_n$ based cosets helps to decode faster. For example, the magic behind the CVP algorithm of $E_8$ \cite{ConwayS82a} is to treat $E_8$ as two $D_8$ cosets while $D_8$ amounts to two $\mathbb{Z}^8$ cosets.

%
%

 \subsubsection{Decoding $BW_{16}$}
 Among $BW_{16}$, $BW_{32}$ and $BW_{64}$, only $BW_{16}$ and $BW_{64}$ contain $D_n$ based cosets:
\begin{align}
BW_{16}  &=  (16,5,8) +2D_{16}, \label{eq_dn50}\\
BW_{64} &= (64,7,32) +2(64,42,8) +4D_{64}.
\end{align}
Their number of cosets are 
$\vert BW_{16}/2D_{16}\vert=2^5$, $\vert BW_{64}/4D_{64}\vert=2^{49}$, contrary to 
$\vert BW_{16}/4\mathbb{Z}^{16}\vert=2^{20}$, $\vert BW_{64}/8\mathbb{Z}^{64}\vert=2^{112}$. In addition, $\vert BW_{32}/4\mathbb{Z}^{32}\vert=2^{32}$.

Summarizing the above, the decoding complexity of $BW_{16}$ seems more affordable than those of $BW_{32}$ and $BW_{64}$.
 With reference to Eqs. (\ref{union_decode1}) and (\ref{eq_dn50}), we have {\begin{align}\label{union_decode2}
	Q_{BW_{16}}(\mathbf{t}) &=  Q_{2 D_{16} + \mathbf{g}'}(\mathbf{t}),\\
	\mathbf{g}' &=\mathrm{argmin}_{\mathbf{g}\in (16,5,8)} \left\| \mathbf{t} -	Q_{2 D_{16}+\mathbf{g}}(\mathbf{t}) \right\|. \nonumber
	\end{align}}
The pseudocode of the CVP algorithms $Q_{BW_{16}}$
and $Q_{D_{n}}$ are listed in Algorithm \ref{AlgL16}
and Algorithm \ref{AlgDn}, respectively.

\begin{algorithm}
	\caption{The closest vector algorithm   $Q_{BW_{16}}$}\label{AlgL16}
	\begin{algorithmic}[1]
		\Require A query vector $\mathbf{y}$.
		\Ensure The closest vector $\hat{\mathbf{v}}$ of $\mathbf{y}$ in $BW_{16}$.
		\State Define the codewords of  $(16,5,8)$ as $\mathbf{d}_1, \ldots, \mathbf{d}_{32}$
			\For{$t=1,\ldots 32$}
		\State	$ \mathbf{y}_t= (\mathbf{y} - \mathbf{d}_t)/2$ \;
		\State	$ \hat{\mathbf{v}}_t = 2 Q_{D_n}({\mathbf{y}_t}) + \mathbf{d}_t$ 
		\State	$\mathrm{Dist}_t = ||\mathbf{y} - \bar{\mathbf{v}}_t || $\;
		\EndFor
	\State	$ t^* = \min_t \mathrm{Dist}_t$ \;
	\State	$\hat{\mathbf{v}} = \hat{\mathbf{v}}_{t^*}$.
	\end{algorithmic}
\end{algorithm}

\begin{algorithm}
	\caption{The closest vector algorithm $Q_{D_{n}}$.}\label{AlgDn}
	\begin{algorithmic}[1]
		\Require A query vector $\mathbf{y}$.
		\Ensure The closest vector $\hat{\mathbf{v}}$ of $\mathbf{y}$ in $D_{n}$.
	\State	$\mathbf{u} = \lfloor \mathbf{y} \rceil$ \;
\State	$\delta = \vert\mathbf{y} - \mathbf{u} \vert$\;
\State	$ t^* = \max_t \vert y_t - u_t \vert $\;
\State	$\mathbf{v} = \mathbf{u}$ \;
	\If{$y_{t^*} - u_{t^*} > 0$}
   \State	{$v_{t^*} \leftarrow v_{t^*} +1$}
	\Else
	\State{$v_{t^*} \leftarrow v_{t^*} -1$}
	\EndIf
	\If{ $u_1+\cdots+u_n \mod 2 = 0$}
\State	{$\hat{\mathbf{v}} = \mathbf{u}$}
	\Else
	\State {$\hat{\mathbf{v}} = \mathbf{v}$}
	\EndIf
	\end{algorithmic}
\end{algorithm}


\section{Improving FrodoPKE with Lattice Codes}

\subsection{DFR Analysis in the Worst Case}


{In FrodoPKE, $\chi_{\sigma}$ is chosen from a truncated discrete Gaussian that minimizes its R\'enyi divergence  from the target “ideal” distribution, as the loss of security can be evaluated by computing the R\'enyi divergence between the two distributions \cite{asiacrypt/Prest17}. To simplify the DFR analysis, $\chi_{\sigma}$ is treated as a continuous Gaussian distribution of $\mathcal{N}(0,\sigma^2)$.}


{Recall that Section 3.1 has formulated an $\bar{m}\bar{n}$-dimensional modulo lattice additive noise channel ``$\mathbf{y}=\mathbf{x}+\mathbf{n}~\mod~q$''.} The error term $\mathbf{n}$ has $\bar{m}\bar{n}$ entries, each entry has the form of $\mathbf{s}'\mathbf{e} + {e}'' -\mathbf{e}'\mathbf{s}$, and we have
\begin{align}
\mathbb{E}(\mathbf{s}'\mathbf{e} + {e}'' -\mathbf{e}'\mathbf{s}) & =0 \\
\mathbb{E}\left(\left\| \mathbf{s}'\mathbf{e} + {e}'' -\mathbf{e}'\mathbf{s} \right\|^2  \right) &= 2n' \sigma^4+\sigma^2. 
\end{align}
Although the entries of $\mathbf{n}$ are not independent, we can use information theory to give a worst case analysis. The information entropy of $\mathbf{n}$ is no larger than that of the joint distribution of  $\bar{m}\bar{n}$ i.i.d. $\mathcal{N}(0,2n' \sigma^4+\sigma^2)$ (also known as Hadamard’s Inequality \cite{cover1999elements}). 
We adopt this ``largest entropy'' setting to approximate the DFR, which amounts to the error rate analysis of lattice codes over an AWGN channel.


{The DFR of the PKE protocol can be estimated by using the 
	decoding error probability $P_e$ of a lattice codeword. To proceed, we set the coarse lattice $\Lambda_c=q\mathbb{Z}^{n}$ ($n=\bar{m}\bar{n}$) as required by the PKE protocol, and identify a general fine lattice $\Lambda_f$ with  kissing number $\tau$, length of the shortest non-zero lattice vector $\lambda_1$, and volume  
	\begin{equation}
	\mathrm{Vol}(\Lambda_f)=\frac{\mathrm{Vol}(\Lambda_c)}{2^{nB}}.
	\end{equation}}


Based on Theorem \ref{thm_er}, the DFR can be evaluated as
\begin{equation}
	P_e \triangleq  \mathrm{Pr}\left(\hat{\mu} \neq \mu\right) =
	\mathrm{Pr}\left( Q_{\Lambda_{f}}(\mathbf{n})
	\notin \Lambda_{c}\right)
\leq \mathrm{Pr}\left( Q_{\Lambda_{f}}(\mathbf{n}) \neq \mathbf{0} \right).
\end{equation}
Assume that $\mathbf{n}$ admits an i.i.d. Gaussian noise $\mathcal{N}(0,\bar{\sigma}^2)$ with $\bar{\sigma}=  \sigma \sqrt{2n'\sigma^2+1}$, it follows from 
\cite[Chap. 3]{Conway1999}, \cite[Eq. 4]{tit/BoutrosVRB96} that 
%
\begin{align}
\mathrm{Pr}\left( Q_{\Lambda_{f}}(\mathbf{n}) \neq \mathbf{0} \right) 
& \lesssim
\frac{\tau}{2} \mathrm{erfc}\left( \frac{\lambda_1/2}{ \sqrt{2} \bar{\sigma}}\right) \label{eq_erratefinal01}\\
& = \frac{\tau}{2} \mathrm{erfc}\left( \frac{\sqrt{\gamma}q}{  2^{B+3/2} \bar{\sigma}}\right), \label{eq_erratefinal}
\end{align}
where the second equality is obtained by substituting $\lambda_1 = \sqrt{\gamma} \left(q^n/2^{nB} \right)^{1/n}$, which is based on the definition of 
Hermite parameter $\gamma$ and $\mathrm{Vol}(\Lambda_c)=q^n$. Note that ``$\lesssim$'' denotes an approximate ``$\leq$'', which holds in the high signal to noise ratio scenario (i.e., $\lambda_1 \gg \bar{\sigma}$) \cite[Chap. 3]{Conway1999}. In Fig. \ref{fig_showEI},  by using $\mathbb{Z}^8$ and $E_8$ as the fine lattice, respectively, we plot both their theoretical DFR upper bounds and the actual simulated DFRs, which suggests the upper bound in (\ref{eq_erratefinal01}) is tight.

\begin{figure}[t!]
	\center
	\includegraphics[width=0.5\textwidth]{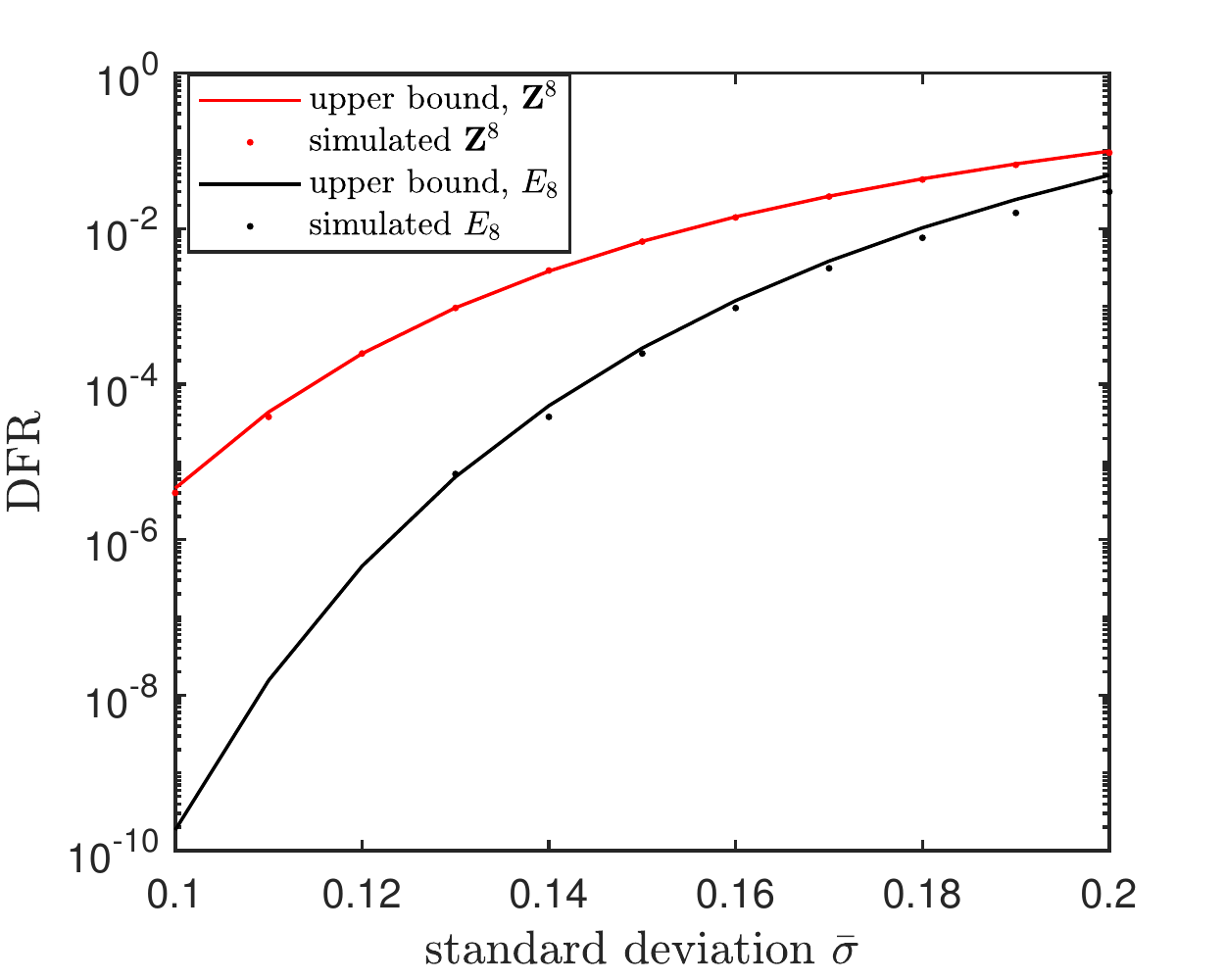}
	\caption{The DFRs of naive modulation and $E_8$ based coded modulation.}
	\label{fig_showEI}
\end{figure}

The DFR formula is determined by a few factors: (i) The Hermite parameter $\gamma$, which describes the density of lattice points packed in a unit volume for a given minimum Euclidean distance. (ii) The kissing number $\tau$ that measure the number of facets in the Voronoi region of a lattice. (iii) The modulus $q$ in LBC. (iv) The averaged number of encoded bits $B$. (v) The standard deviation $\bar{\sigma}$ of the effective noise.


\subsection{Flexible Lattice Parameter Settings}

Finding the densest lattice structure is a well-studied topic, and the Hermite parameter $\gamma$ and kissing number $\tau$ of some low-dimensional optimal lattices can be found in \cite{Conway1999}. Therefore, the key challenge is to judiciously design $B$, $q$, $\bar{\sigma}$ based on chosen $\gamma$ and $\tau$.

\noindent  \textit{i) On the kissing number and Hermite parameter.}
We adopt Barnes--Wall lattices to construct lattice codes. 
Though being less dense than other known packings in dimensions $32$ and higher,  they offer the densest packings in dimensions $2$, $4$, $8$ and $16$ \cite{Conway1999}. Moreover, many lattice parameters are available \cite{Conway1999}[P. 151]. In dimension $n=2^r$ with $r = 1, 2, 3, \ldots$, the kissing number is
\begin{equation}\label{eq_kissing}
\tau = (2+2)(2+2^2)\cdots (2+2^r),
\end{equation}
and the Hermite parameter is
\begin{equation}
\gamma_{r}= 2^{(r-1)/2},
\end{equation}
which increases without limit. 	If $\Lambda'$ is constructed from the $k$-fold Cartesian product of $\Lambda \subset \mathbb{R}^m$, i.e, $\Lambda'=\Lambda\times\cdots \times \Lambda \subset \mathbb{R}^{km}$, then we have
\begin{align}
\tau(\Lambda') &= k	\tau(\Lambda) \\
\gamma_{r}(\Lambda') &= \gamma_{r}(\Lambda).
\end{align}

Table \ref{tab_kissing} summarizes the parameters of some low-dimensional optimal lattices and the Barnes--Wall lattices. 
 
\begin{table}[th!]
	\begin{center}
		\begin{minipage}{1\textwidth}
			\caption{The properties of some popular lattices.}\label{tab_kissing}
			\begin{tabular*}{1\textwidth}{p{0.18\textwidth}|c|c|c|c|c|c|c}
				\hline 
				\multirow{2}{*}{ } & \multirow{2}{*}{$\text{\ensuremath{\mathbb{Z}}}$} & \multirow{2}{*}{$D_{4}$} & \multirow{2}{*}{$E_{8}$} & \multirow{2}{*}{$BW_{16}$} & \multirow{2}{*}{$\Lambda_{24}$} & \multirow{2}{*}{$BW_{32}$} & \multirow{2}{*}{$BW_{64}$}\tabularnewline
				&  &  &  &  &  &  & \tabularnewline
				\hline 
				\hline 
				Hermite parameter $\gamma$ & $1$ & $2^{1/2}$ & $2$ & $2^{3/2}$ & $4$ & $4$ & $2^{5/2}$\tabularnewline
				{Kissing number} $\tau$ & $2$ & $24$ & $240$ & $4320$ & $196560$ & $146880$ & $9694080$\tabularnewline
				Volume $\mathrm{Vol}(\Lambda_t)$ & $1$ & $2$ & $1$ & $2^{12}$ & $1$ & $2^{32}$ & $2^{80}$\tabularnewline
				\hline 
			\end{tabular*}
		\end{minipage}
	\end{center}
\end{table}
 

\noindent  \textit{ii) On the information rate $B$.}
Since the coarse lattice in FrodoPKE is $\Lambda_c=q \mathbb{Z}^{64}$, let $2^{\Delta}=q/p$ be a power of $2$ with $p$ being a free parameter. 
By choosing a small dimensional lattice   $\Lambda_t \in \mathbb{R}^t$, $t$ dividing $n=64$,
 and  $p\mathbb{Z}^t   \subset \Lambda_t$,   the fine lattice is a Cartesian product of  $\Lambda_t$:
 \begin{equation}
 \Lambda_f=	2^{\Delta} \Lambda_t \times \cdots \times \Lambda_t.
 \end{equation}
 Then the number of encoded bits $B$ per dimension is dictated by $p$:
 \begin{align} \label{eq_rate1}
 B & =  \frac{1}{n} \log_2 \left(   \frac{\mathrm{Vol}(\Lambda_c)}{\mathrm{Vol}(\Lambda_f)}    \right)   = \frac{1}{t} \log_2  \frac{p^t}{ \mathrm{Vol}(\Lambda_t)}. 
 \end{align}
For a Construction-A or Construction-D lattice, one always has
\begin{equation}
	p\mathbb{Z}^t \subset 2^a \mathbb{Z}^t \subset \Lambda_t.
\end{equation}
While the $E_8$ lattice has half integers, it holds that $4\mathbb{Z}^8 \subset 2E_8$. 


Based on different fine lattices, 
 we enumerate some feasible number of encoded bits in FrodoPKE below, denoted as $64B$.
 \begin{itemize}
 	\item $\Lambda_f=2^{\Delta}\cdot \mathbb{Z}^{64}$, $64B=64, 128, 192, 256, \ldots$
 	\item $\Lambda_f=2^{\Delta}\cdot D_4^{16}$, $64B=112, 176, 240, 304, \ldots$
 	\item $\Lambda_f=2^{\Delta}\cdot E_8^8$, $64B=64, 128, 192, 256, \ldots$
 		 	\item $\Lambda_f=2^{\Delta}\cdot  BW_8^{8}$, $64B=96, 160, 224, 288, \ldots$
 	\item $\Lambda_f=2^{\Delta}\cdot BW_{16}^4$, $64B=80, 144, 208, 272, \ldots$
 	\item $\Lambda_f=2^{\Delta}\cdot BW_{32}^2$, $64B=64, 128, 192, 256, \ldots$
 	\item $\Lambda_f=2^{\Delta}\cdot BW_{64}$, $64B=112, 176, 240, 304 \ldots$
 \end{itemize}
 



\subsection{Improved Frodo Parameters}
Frodo-640, Frodo-976 and Frodo-1344 target security levels 1, 3 and 5 in the NIST PQC Standardization, respectively. To resist the attack exploiting DFRs \cite{impact_dfr_DAnversVV18}, the DFRs at levels 1, 3 and 5 should be no larger than $2^{-128}$, $2^{-192}$ and $2^{-256}$, respectively.

Compared to the standard Frodo protocol, 
our scheme only modifies the labeling function, the corresponding CVP algorithm, and the choice of parameters $\sigma, B, q$. The security levels  refer to the primal and dual attack
via the FrodoKEM script $\mathsf{ pqsec.py}$ \cite{pqsec2016}.
The subscripts C, Q and P denote ``classical'', ``quantum'' and ``paranoid'' estimates on the concrete bit-security given by parameters ($n', \sigma, q$).
We propose three sets of parameters in Tables 2 and 3: the first aims at improving the security level and the second at reducing the communication bandwidth. Frodo-640/976/1344 are the original parameter sets. The 
parameters that we have changed are highlighted in bold-face blue color, and 
other values that have altered as a consequence of this change are marked with normal blue color.

\noindent\textbf{Parameter set 1: Improved security strength}

We increase $\sigma$ while keeping $n', q$ unchanged in Frodo-640/976/1344. As shown in Table \ref{tab_GcDimN2}, error correction via $E_8$, $BW_{16}$ and $BW_{32}$ can improve the security level of the original Frodo-640/976/1344 by $6$ to $16$ bits. While $\mathbb{Z}^{64}$, $E_{8}^8$ and $BW_{32}^2$ can naturally encode $128$, $192$ and $256$ bits per instance,  $BW_{16}^4$ only supports $144$, $208$, and $272$ bits. The
  $BW_{32}$ based parameter set offers the highest security enhancement in the table, but its CVP decoding complexity of $O(2^{32})$ makes it less attractive. 

We recommend the $E_8$ and $BW_{16}$ based parameter sets. The information rate of 	Frodo-640/976/1344-$E_{8}$ matches well with that of the original Frodo-640/976/1344, and the classical security level has been increased by $7$ or $8$ bits, respectively. Frodo-640/976/1344-$BW_{16}$ maintains basically the same security level  as that of Frodo-640/976/1344-$E_{8}$, while the information rate is slightly higher, either $B=2.25$, $3.25$ or $4.25$.
 
\noindent\textbf{Parameter set 2: Reduced size of ciphertext}

Recall that the size of ciphertext is $(\bar{m}n'+\bar{m}\bar{n})\log_2(q)/8$ bytes, so we  reduce $q$ to achieve higher bandwidth efficiency. To keep the DFR small, we also reduce $\sigma$ to various degrees, as long as the security level is no smaller.

As shown in Table \ref{tab_GcDimN3}, by reducing $q$ from $2^{15}$ to $2^{14}$, the ciphertext size $|c|$ can be reduced from $9720$ bytes to $9072$ bytes in Frodo-640, from $15744$ bytes to $14760$ bytes in Frodo-976, and from $21632$ bytes to $20280$ bytes in Frodo-1344. Again, the $E_8$ and $BW_{16}$ based parameter sets are recommended.

{It is interesting to note that the lattice-code based FrodoPKE can also be extended to a KEM for symmetric lightweight cryptography algorithms. For instance, via setting $\Lambda_f= 2^{\Delta}\cdot BW_{16}^4, \Lambda_c= 2^{\Delta}\cdot 4 \mathbb{Z}^{64}$, it is possible to tightly exchange $80$ bits for the PRESENT algorithm \cite{access/ThakorRK21}.}

\subsection{IND-CCA Security}
The lattice codes based PKE/KEM also features chosen ciphertext secure (IND-CCA) security.
Similarly to  the argument in \cite{naehrig2017frodokem}, the IND-CPA security of FrodoPKE is upper bounded by the advantage of the decision-LWE problem for the same parameters and error distribution. 
To endow an IND-CPA encryption scheme with IND-CCA security, the post-quantum secure version of the Fujisaki-Okamoto transform \cite{DBLP:conf/crypto/FujisakiO99,DBLP:conf/tcc/HofheinzHK17} can be applied. {When bounding the probability that an attacker can undermine a given cryptographic scheme in the quantum random-oracle model, 
	security proofs use the number of decryption queries submitted by the CCA adversary. \cite[Theorem 4.3]{frodothm2023} shows that the impact of decryption failure is given by $4q_G P_e$ where $q_G$ is the number of quantum oracle queries and $P_e$ is the DFR. Then, using
	bounds on decryption failure established above, one can argue that such queries pose no danger.}

\begin{sidewaystable}
	\sidewaystablefn%
	\begin{center}
		\begin{minipage}{\textheight}
			\caption{The recommended parameter sets with higher security.}\label{tab_GcDimN2}
		\begin{tabular}{c|c|c|c|c|c|c|c|c|c|c|c}
		\hline 
		\multirow{2}{*}{ } & \multicolumn{2}{c|}{Structure of lattice code} & \multirow{2}{*}{$n',\bar{n},\bar{m}$} & \multirow{2}{*}{$q$} & \multirow{2}{*}{$\sigma$} & \multirow{2}{*}{$B$} & \multirow{2}{*}{DFR} & \multicolumn{1}{c|}{$c$ size} & \multicolumn{3}{c}{Security}\tabularnewline
		\cline{2-3} \cline{10-12} 
		& $\Lambda_{f}$ & $\Lambda_{c}$ &  &  &  &  &  & (bytes) & C & Q & P\tabularnewline
		\hline 
		Frodo-640 & $2^{13}\cdot\mathbb{Z}^{64}$ & $2^{15}\cdot \mathbb{Z}^{64}$ & $640,8,8$ & $2^{15}$ & $2.75$ & $2$ & $2^{-164}$ & $9720$ & $149$ & $136$ & $109$\tabularnewline
		Frodo-640-$E_{8}$ & $2^{13}\cdot E_{8}^{8}$ & $2^{15}\cdot \mathbb{Z}^{64}$ & $640,8,8$ & $2^{15}$ & \textcolor{blue}{ $\mathbf{3.25}$} & $2$ & $2^{-164}$ & $9720$ & \textcolor{blue}{$156$} & \textcolor{blue}{$142$} & \textcolor{blue}{$113$}\tabularnewline
		Frodo-640-$BW_{16}$ & $2^{12}\cdot BW_{16}^{4}$ & $2^{15}\cdot \mathbb{Z}^{64}$ & $640,8,8$ & $2^{15}$ & \textcolor{blue}{$\mathbf{3.23}$} & \textcolor{blue}{$\mathbf{2.25}$} & 
		$2^{-164}$ & $9720$ & \textcolor{blue}{$155$} & \textcolor{blue}{$142$} & \textcolor{blue}{$113$}\tabularnewline
		Frodo-640-$BW_{32}$ & $2^{12}\cdot BW_{32}^2$ & $2^{15}\cdot \mathbb{Z}^{64}$ & $640,8,8$ & $2^{15}$ & \textcolor{blue}{$\mathbf{3.83}$} &  {$2$} & $2^{-164}$ & $9720$ & \textcolor{blue}{$162$} & \textcolor{blue}{$148$} & \textcolor{blue}{$118$}\tabularnewline
		\hline 
		Frodo-976 & $2^{13}\cdot\mathbb{Z}^{64}$ & $2^{16}\cdot \mathbb{Z}^{64}$ & $976,8,8$ & $2^{16}$ & $2.3$ & $3$ & $2^{-220}$ & $15744$ & $216$ & $196$ & $156$\tabularnewline
		Frodo-976-$E_{8}$ & $2^{13}\cdot E_{8}^{8}$ & $2^{16}\cdot \mathbb{Z}^{64}$ & $976,8,8$ & $2^{16}$ & 
		\textcolor{blue}{$\mathbf{2.72}$} & $3$ & $2^{-220}$ & $15744$ & \textcolor{blue}{$224$} & \textcolor{blue}{$204$} & \textcolor{blue}{$162$}\tabularnewline
		Frodo-976-$BW_{16}$ & $2^{12}\cdot BW_{16}^{4}$ & $2^{16}\cdot \mathbb{Z}^{64}$ & $976,8,8$ & $2^{16}$ & \textcolor{blue}{$\mathbf{2.71}$} & \textcolor{blue}{$\mathbf{3.25}$} & 
		$2^{-220}$ & $15744$ & \textcolor{blue}{$224$} & \textcolor{blue}{$204$} & \textcolor{blue}{$161$}\tabularnewline
		Frodo-976-$BW_{32}$ & $2^{12}\cdot BW_{32}^2$ & $2^{16}\cdot \mathbb{Z}^{64}$ & $976,8,8$ & $2^{16}$ & \textcolor{blue}{$\mathbf{3.21}$} &  
		{$3$} & $2^{-220}$ & $15744$ & \textcolor{blue}{$232$} & \textcolor{blue}{$211$} & \textcolor{blue}{$167$}\tabularnewline
		\hline 
		Frodo-1344 & $2^{12}\cdot\mathbb{Z}^{64}$ & $2^{16}\cdot \mathbb{Z}^{64}$ & $1344,8,8$ & $2^{16}$ & $1.4$ &
		$4$ & $2^{-290}$ & $21632$ & 
		$282$ & $256$ & $203$\tabularnewline
		Frodo-1344-$E_{8}$ & $2^{12}\cdot E_{8}^{8}$ & $2^{16}\cdot \mathbb{Z}^{64}$ & $1344,8,8$ & $2^{16}$ & 
		\textcolor{blue}{$\mathbf{1.66}$} & 
		$4$ & $2^{-290}$ & $21632$ & \textcolor{blue}{$292$} & \textcolor{blue}{$265$} & \textcolor{blue}{$210$}\tabularnewline
		Frodo-1344-$BW_{16}$ & $2^{11}\cdot BW_{16}^{4}$ & $2^{16}\cdot \mathbb{Z}^{64}$ & $1344,8,8$ & $2^{16}$ & \textcolor{blue}{$\mathbf{1.66}$} & \textcolor{blue}{$\mathbf{4.25}$} & 
		$2^{-290}$ & $21632$ & \textcolor{blue}{$292$} & \textcolor{blue}{$265$} & \textcolor{blue}{$210$}\tabularnewline
		Frodo-1344-$BW_{32}$ & $2^{11}\cdot BW_{32}^2$ & $2^{16}\cdot \mathbb{Z}^{64}$ & $1344,8,8$ & $2^{16}$ & \textcolor{blue}{$\mathbf{1.97}$} &  
		$4$ & $2^{-290}$ & $21632$ & \textcolor{blue}{$302$} & \textcolor{blue}{$275$} & \textcolor{blue}{$217$}\tabularnewline
		\hline 
	\end{tabular}
		\end{minipage}
	\end{center}
\end{sidewaystable}

\begin{sidewaystable}
	\sidewaystablefn%
	\begin{center}
		\begin{minipage}{\textheight}
			\caption{The recommended parameter sets with smaller size of ciphertext.}\label{tab_GcDimN3}
	\begin{tabular}{c|c|c|c|c|c|c|c|c|c|c|c}
	\hline 
	\multirow{2}{*}{ } & \multicolumn{2}{c|}{Structure of lattice code} & \multirow{2}{*}{$n',\bar{n},\bar{m}$} & \multirow{2}{*}{$q$} & \multirow{2}{*}{$\sigma$} & \multirow{2}{*}{$B$} & \multirow{2}{*}{DFR} & \multicolumn{1}{c|}{$c$ size} & \multicolumn{3}{c}{Security}\tabularnewline
	\cline{2-3} \cline{10-12} 
	& $\Lambda_{f}$ & $\Lambda_{c}$ &  &  &  &  &  & (bytes) & C & Q & P\tabularnewline
	\hline 
	Frodo-640 & $2^{13}\cdot\mathbb{Z}^{64}$ & $2^{15}\cdot \mathbb{Z}^{64}$ & $640,8,8$ & $2^{15}$ & $2.75$ & $2$ & $2^{-164}$ & $9720$ & $149$ & $136$ & $109$\tabularnewline
	Frodo-640-$E_{8}$ & $2^{12}\cdot E_{8}^{8}$ & $2^{14}\cdot \mathbb{Z}^{64}$ & $640,8,8$ & \textcolor{blue}{$\mathbf{2^{14}}$} & \textcolor{blue}{$2.30$} & $2$ &  {$2^{-164}$} & \textcolor{blue}{$9072$} & \textcolor{blue}{$156$} & \textcolor{blue}{$143$} & \textcolor{blue}{$114$}\tabularnewline
	Frodo-640-$BW_{16}$ & $2^{11}\cdot BW_{16}^{4}$ & $2^{14}\cdot \mathbb{Z}^{64}$ & $640,8,8$ & \textcolor{blue}{$\mathbf{2^{14}}$} & \textcolor{blue}{$2.29$} & \textcolor{blue}{$\mathbf{2.25}$} &  {$2^{-164}$} & \textcolor{blue}{$9072$} & \textcolor{blue}{$156$} & \textcolor{blue}{$143$} & \textcolor{blue}{$114$}\tabularnewline
	Frodo-640-$BW_{32}$ & $2^{11}\cdot BW_{32}^2$ & $2^{14}\cdot \mathbb{Z}^{64}$ & $640,8,8$ & \textcolor{blue}{$\mathbf{2^{14}}$} & \textcolor{blue}{$2.71$} &  {$2$} &  {$2^{-164}$} & \textcolor{blue}{$9072$} & \textcolor{blue}{$163$} & \textcolor{blue}{$149$} & \textcolor{blue}{$118$}\tabularnewline
	\hline 
	Frodo-976 & $2^{13}\cdot\mathbb{Z}^{64}$ & $2^{16}\cdot \mathbb{Z}^{64}$ & $976,8,8$ & $2^{16}$ & $2.3$ & $3$ & $2^{-220}$ & $15744$ & $216$ & $196$ & $156$\tabularnewline
	Frodo-976-$E_{8}$ & $2^{12}\cdot E_{8}^{8}$ & $2^{15}\cdot \mathbb{Z}^{64}$ & $976,8,8$ & \textcolor{blue}{$\mathbf{2^{15}}$} & \textcolor{blue}{$1.93$} & $3$ &  {$2^{-220}$} & \textcolor{blue}{$14760$} & \textcolor{blue}{$225$} & \textcolor{blue}{$205$} & \textcolor{blue}{$162$}\tabularnewline
	Frodo-976-$BW_{16}$ & $2^{11}\cdot BW_{16}^{4}$ & $2^{15}\cdot \mathbb{Z}^{64}$ & $976,8,8$ & \textcolor{blue}{$\mathbf{2^{15}}$} & \textcolor{blue}{1.92} & \textcolor{blue}{$\mathbf{3.25}$} &  
	{$2^{-220}$} & 
	\textcolor{blue}{$14760$} & \textcolor{blue}{$224$} & \textcolor{blue}{$204$} & \textcolor{blue}{$162$}\tabularnewline
	Frodo-976-$BW_{32}$ & $2^{11}\cdot BW_{32}^2$ & $2^{15}\cdot \mathbb{Z}^{64}$ & $976,8,8$ & \textcolor{blue}{$\mathbf{2^{15}}$} & \textcolor{blue}{$2.27$} &  {$3$} & {$2^{-220}$} & 
	\textcolor{blue}{$14760$} & \textcolor{blue}{$233$} & \textcolor{blue}{$212$} & \textcolor{blue}{$168$}\tabularnewline
	\hline 
	Frodo-1344 & $2^{12}\cdot\mathbb{Z}^{64}$ & $2^{16}\cdot \mathbb{Z}^{64}$ & $1344,8,8$ & $2^{16}$ & $1.4$ &
	$4$ & $2^{-290}$ & $21632$ & 
	$282$ & $256$ & $203$\tabularnewline
	Frodo-1344-$E_{8}$ & $2^{11}\cdot E_{8}^{8}$ & $2^{15}\cdot \mathbb{Z}^{64}$ & $1344,8,8$ & \textcolor{blue}{$\mathbf{2^{15}}$} & 
	\textcolor{blue}{${1.18}$} & 
	$4$ & $2^{-290}$ & \textcolor{blue}{$20280$} & \textcolor{blue}{$291$} & \textcolor{blue}{$265$} & \textcolor{blue}{$210$}\tabularnewline
	Frodo-1344-$BW_{16}$ & $2^{10}\cdot BW_{16}^{4}$ & $2^{15}\cdot \mathbb{Z}^{64}$ & $1344,8,8$ & \textcolor{blue}{$\mathbf{2^{15}}$} & \textcolor{blue}{${1.17}$} & \textcolor{blue}{$\mathbf{4.25}$} & 
	$2^{-290}$ & \textcolor{blue}{$20280$} & \textcolor{blue}{$291$} & \textcolor{blue}{$265$} & \textcolor{blue}{$209$}\tabularnewline
	Frodo-1344-$BW_{32}$ & $2^{10}\cdot BW_{32}^2$ & $2^{15}\cdot \mathbb{Z}^{64}$ & $1344,8,8$ & \textcolor{blue}{$\mathbf{2^{15}}$} & \textcolor{blue}{${1.39}$} &  
	$4$ & $2^{-290}$ & \textcolor{blue}{$20280$} & \textcolor{blue}{$302$} & \textcolor{blue}{$275$} & \textcolor{blue}{$217$}\tabularnewline
	\hline 
\end{tabular}
		\end{minipage}
	\end{center}
\end{sidewaystable}

\section{Conclusions}

While the cryptography community is more familiar with random lattices for security, this paper shows that low-dimensional structure lattices can improve the error correction performance in FrodoPKE. The rationale is that lattice codes represent coded modulation, the elegant combination of ECC and modulation. 
The bridge that connects lattice codes and FrodoPKE (and more generally lattice-based PKEs) is the modulo $q$ operation, which induces hypercube shaping. By presenting an efficient lattice labeling function, as well as a general formula to estimate the DFR, lattice based coded modulation becomes practical in LBC.  By using some low-dimensional optimal lattices, a few improved parameter sets for FrodoPKE have been achieved, with either higher security or smaller ciphertext sizes.
The lattice coding techniques in this work can be similarly 
applied to Ring/Module LWE-based PKEs.

\section*{Appendix A}
The lattice bases of  $E_{8}$, $BW_{8}$  and $BW_{16}$ can be respectively chosen as 
\begin{align*}
&\left[\begin{array}{cccccccc}
2  &  -1  &   0  &   0   &  0   &  0   &  0  &   0.5\\
0 &    1   & -1  &   0  &   0  &   0  &   0  &   0.5\\
0  &   0 &    1 &   -1  &   0 &    0   &  0  &   0.5\\
0 &    0  &   0 &    1  &  -1  &   0   &  0  &   0.5\\
0 &    0  &   0  &   0  &   1  &  -1   &  0  &   0.5\\
0 &    0  &   0  &   0  &   0  &   1  &  -1  &   0.5\\
0 &    0  &   0  &   0  &   0  &   0 &    1  &   0.5\\
0 &    0  &   0  &   0  &   0  &   0 &    0  &   0.5
\end{array}\right],
 \, \, 
 \left[\begin{array}{cccccccc}
1   & 1   &  1 &   1  &   2 &    2   &  2   &  2 \\
1  &   1  &   1  &   0 &    2  &   0  &   0 &    0\\
1  &   1 &    0  &   1 &    0 &   2   &  0   &  0\\
1  &   1 &    0  &   0 &    0  &   0  &   0  &   0\\
1  &   0  &   1  &   1 &    0 &    0  &   2  &   0\\
1  &   0 &    1  &   0 &    0 &    0  &   0  &   0\\
1  &   0  &   0  &   1 &    0  &   0  &   0  &   0\\
1  &   0 &    0  &   0 &    0  &   0  &   0  &   0
\end{array}\right]\\
&
\left[\begin{array}{cccccccccccccccc}
1&1&1&1&1&2&2&2&2&2&2&2&2&2&2&4\\
1&1&1&1&0&2&2&0&2&0&0&2&0&0&0&0\\
1&1&1&0&1&2&0&2&0&2&0&0&2&0&0&0\\
1&1&1&0&0&2&0&0&0&0&0&0&0&0&0&0\\
1&1&0&1&1&0&2&2&0&0&2&0&0&2&0&0\\
1&1&0&1&0&0&2&0&0&0&0&0&0&0&0&0\\
1&1&0&0&1&0&0&2&0&0&0&0&0&0&0&0\\
1&1&0&0&0&0&0&0&0&0&0&0&0&0&0&0\\
1&0&1&1&1&0&0&0&2&2&2&0&0&0&2&0\\
1&0&1&1&0&0&0&0&2&0&0&0&0&0&0&0\\
1&0&1&0&1&0&0&0&0&2&0&0&0&0&0&0\\
1&0&1&0&0&0&0&0&0&0&0&0&0&0&0&0\\
1&0&0&1&1&0&0&0&0&0&2&0&0&0&0&0\\
1&0&0&1&0&0&0&0&0&0&0&0&0&0&0&0\\
1&0&0&0&1&0&0&0&0&0&0&0&0&0&0&0\\
1&0&0&0&0&0&0&0&0&0&0&0&0&0&0&0
\end{array}\right].
\end{align*}

\bibliography{lib}


\end{document}